\documentclass{elsarticle}

\usepackage{fullpage}
\usepackage{latexsym}
\usepackage{epsfig}
\usepackage{color}

\newtheorem{theorem}{Theorem}
\newtheorem{lemma}{Lemma}

\newtheorem{corollary}{Corollary}
\newenvironment{proof}{{\bf Proof: }}{\hspace*{\fill}$\Box$\medskip}

\begin{document}

\title{Transforming planar graph drawings while maintaining height
}
\author[uw]{Therese Biedl}
\address[uw]{David R.~Cheriton School of Computer Science, 
University of Waterloo, 
Waterloo, ON N2L 3G1, Canada, {\tt biedl@uwaterloo.ca}}
\date{\today}

\begin{abstract}
There are numerous styles of planar graph drawings,  notably
straight-line drawings, poly-line drawings, orthogonal graph
drawings and visibility representations.  In this note,
we show that many of these drawings can be transformed from
one style to another without changing the height of the drawing.
We then give some applications of these transformations. 
\end{abstract}


\maketitle
\section{Introduction}
\label{se:intro}

Let $G=(V,E)$ be a simple graph with $n=|V|$ vertices and
$m=|E|$ edges.  We assume that $G$ is {\em planar},
i.e., it can be drawn without crossing.    
In planar graph drawing, one aims to create a crossing-free picture of $G$.
It was known for a long time that such drawings
exist even with straight lines \cite{Wagner36,Fary48,Stein51}.
and even in an $O(n)\times O(n)$-grid \cite{FPP90,Sch90}.
Many improvements have been developed since, see for example
\cite{DBETT98,NR04}.

Formally, 
a {\em drawing} of a graph consists of assigning a point or an axis-aligned
box to every vertex, and a curve between the points/boxes of  $u$ and
$v$ to every edge $(u,v)$.   The drawing is called {\em planar} if
no two elements of the drawing intersect unless the corresponding elements
of the graph do.  Thus no 
vertex points/boxes coincide, no edge curve
self-intersects, no edge curves intersect each other (except at common
endpoints), and no edge curve intersects a vertex point/box other than 
its endpoints.  In this paper all drawings
are required to be planar.  

In a {\em straight-line drawing} vertices are
represented by points and edges are drawn as straight-line segments.
In a {\em poly-line drawing}, vertices are points and each edge curve is a
continguous sequence of line segments.  The place where an edge curves
changes direction is called a {\em bend}.
An {\em orthogonal drawing} uses a box for every vertex, and
requires that edges are drawn as poly-lines for which every line
segment is either horizontal or vertical.%
\footnote{If the maximum degree is 4, then one could additionally
demand vertices to be points.  We will not study this model in the
current paper, so an orthogonal drawing always allows boxes.}
A {\em visibility
representations} is an orthogonal drawing without bends.

In this paper we sometimes restrict these drawings further.
An orthogonal drawing is called
{\em flat} if all boxes of vertices are degenerated into horizontal
segments.
We say that a drawing $\Gamma$ is {\em $y$-monotone} if for any edge
$(v,w)$, the drawing of the edge in $\Gamma$ forms a $y$-monotone
path.  (Horizontal edge segments are allowed.)
Any straight-line drawing
and any visibility representation is automatically $y$-monotone.
See also Fig.~\ref{fi:draw_ex}.  (In our drawings, we thicken vertex boxes
slightly, so that horizontal segments appear as boxes of small height.)

In all drawings, the defining elements (i.e.,
points of vertices, corners of boxes of vertices, bends, and attachment points
of edges to vertex-boxes) must be placed
at points with integer coordinates.  A drawing is said to have
{\em width $w$} and {\em height $h$} if (possibly after translation)
all such points are placed on the $[1,w]\times [1,h]$-grid.  
The height is thus measured by the number of {\em rows}, 
i.e., horizontal lines with integer $y$-coordinates that are occupied by the 
drawing.  After rotation, we may assume that the height is no larger
than the width.

In a previous paper \cite{Bie-DCG11}, we studied transformations
between these graph drawing styles that preserved the asymptotic
area.  In particular, we showed that if $G$ has a visibility
representation, then it has a poly-line drawing of asymptotically
the same area.  In this paper, we study such
transformations with the goal of keeping
the height of the drawing unchanged.
Our results are illustrated in Figure~\ref{fig:transform_graph}.
All our transformations have two additional properties that 
are useful in the proofs and 
some of the applications.  First, the resulting drawing
{\em has the same
 $y$-coordinates}, i.e., any vertex (and also any bend that is
not removed) has the same $y$-coordinate in
the new drawing as in the original one.  (Since we give transformations
only for flat orthogonal drawings, this concept makes sense even when
transforming vertex-boxes into points.) Second, the resulting drawings
have {\em the same left-to-right order in each row}, i.e., if 
vertices $v$ and $w$
had the same $y$-coordinate, with $v$ left of $w$, then the same
also holds in the resulting drawing.

\begin{figure}[ht]
\hspace*{\fill}
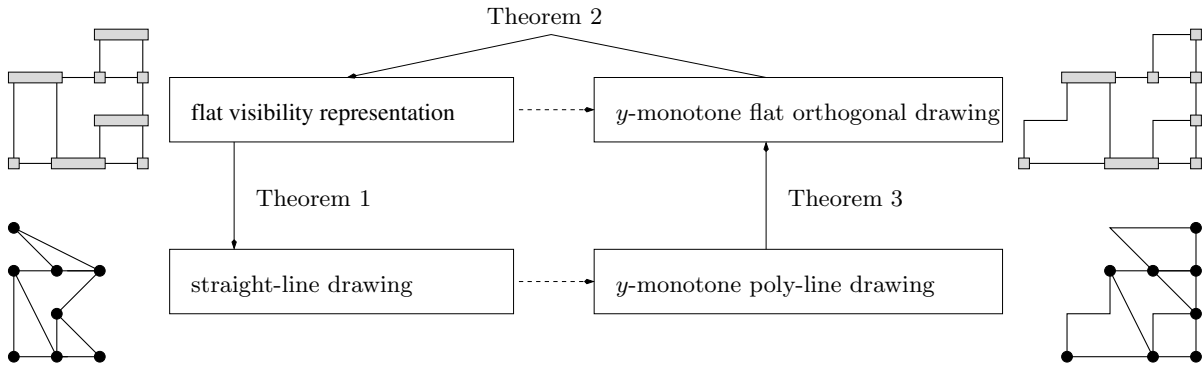
\hspace*{\fill}
\caption{Height-preserving transformations proved in
this paper. Dashed arrows are trivial implications.  
}
\label{fig:transform_graph}
\label{fi:draw_ex}
\end{figure}

We then study some applications of these results.  Most importantly,
they allow to derive some height-bounds for drawing styles for
which we are not aware of any direct proof, and they allow to
formulate some NP-hard graph drawing problems as integer programs.

\section{Flat visbility representations to straight-line drawings}
\label{se:VR_SL}

\begin{theorem}
\label{thm:VR2SL}
\label{thm:VR_SL}
\label{th:VR2SL}
Any flat visibility representation $\Gamma$ can be transformed into a 
straight-line drawing $\Gamma'$ with the same $y$-coordinates
and the same left-to-right orders in each row.
\end{theorem}
\begin{proof}
For any vertex $v$, use $x_l(v), x_r(v)$ and $y(v)$ to denote leftmost and
rightmost $x$-coordinate and (unique) $y$-coordinate of the box that 
represents $v$ in $\Gamma$.   
Use $X(v)$ and $Y(v)$ to denote
the (initially unknown) coordinates of $v$ in $\Gamma'$. 
For any vertex set $Y(v)=y(v)$, hence $y$-coordinates are the same.

Let $v_1,\dots,v_n$ be the vertices sorted by $x_l(.)$, breaking ties
arbitrarily. The algorithm determines $X(.)$ for each vertex by 
processing vertices in 
this order and expanding the drawing $\Gamma'_{i-1}$ created for 
$v_1,\dots,v_{i-1}$ into a drawing $\Gamma'_i$ of $v_1,\dots,v_i$,
which has the same left-to-right orders. 

Suppose $X(v_g)$ has been computed for all $g<i$ already.  To find
$X(v_i)$, determine lower bounds for it by considering all predecessors
of $v_i$ and taking the maximum over all of them.  
(For each vertex $v_i$, the {\em predecessors}
of $v_i$ are the neighbours of $v_i$ that come earlier in the order
$v_1,\dots,v_n$.) A first (trivial)
lower bound for $X(v_i)$ is that it needs to be to the right of
anything in row $y(v_i)$.  Thus, if $\Gamma'_{i-1}$
contains a vertex
or part of an edge at point $(X,y(v_i))$, then $X(v_i)\geq 
\lfloor X \rfloor + 1$ is required.

Next consider any predecessor $v_g$ of $v_i$ with $y(v_g)\neq y(v_i)$.  
Since $v_g$ and $v_i$ are not in the same row, they must see each other
vertically in $\Gamma$, which means that $x_r(v_g)\geq x_l(v_i)$.
See also Fig.~\ref{fig:transform}.
So if $v_g$ has a neighbour $v_k$ to its right in $\Gamma$, then
$x_\ell(v_k)> x_r(v_g) \geq x_\ell(v_i)$, which implies that 
$v_k$ has not yet been added to $\Gamma'_{i-1}$.  Since $\Gamma_{i-1}'$
has the same left-to-right orders, $v_g$ is hence the rightmost
vertex in its row in $\Gamma'_{i-1}$ and can see towards infinity on the
right.  But then $v_g$ can also see the point $(+\infty,y(v_i))$, or
in other words, there exists some $X_g$ such that $v_g$ can see all
points $(X,y(v_i))$ for $X\geq X_g$.  
Impose the lower bound $X(v_i) \geq \lfloor X_g \rfloor + 1$ on the
$x$-coordinate of $v_i$.

\begin{figure}[ht]
\hspace*{\fill}
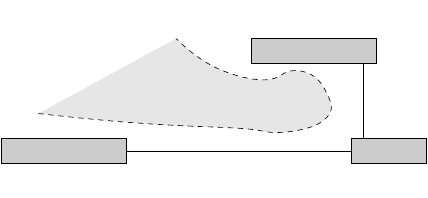
\hspace*{\fill}
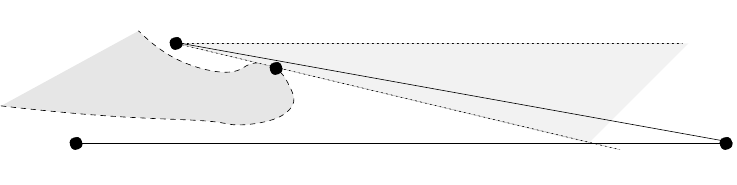
\hspace*{\fill}
\caption{Transforming a flat visibility drawing into a straight-line
drawing with the same $y$-coordinates.}
\label{fig:transform}
\end{figure}

Now let $X(v_i)$ be the smallest value that satisfies the above lower
bounds (from the row $y(v_i)$ and from all predecessors of $v_i$ in
different rows.)  Set $X(v_i)=0$ if there were no such lower bounds.%
\footnote{To simplify the calculations below it helps to use 0 (as
opposed to 1) for the leftmost $x$-coordinate.}
Directly by construction, placing $v_i$ at
$(X(v_i),y(v_i))$ allows it to be connected with straight-line segments
to all its predecessors.  This includes the predecessor (if any)
that is in the same row as $v_i$, since there can be at most one in a flat
visibility representation, and it is in the same row as $v_i$.
This gives a drawing $\Gamma'_i$ of $v_1,\dots,v_i$ as
desired, and the result follows by induction.
\end{proof}

\subsection{Width considerations}

While our transformation keeps the height intact,
the width can increase dramatically.  
Fig.~\ref{fig:bad_example} shows a flat visibility representations of 
height $h$ and width $O(n)$ 
such that the transformation of Theorem~\ref{th:VR2SL} 
has width $\Omega((h-2)^{(n-3)})$.
Specifically, using induction one shows that vertex $i$ (for
$i\geq 3$) is placed
with $x$-coordinate $1+(h-2)+\dots+(h-2)^{i-3}$  and leaves an 
edge with slope $\pm 1 / (1+(h-2)+\dots+(h-2)^{i-3})$.
But this is (asymptotically) the
worst that can happen.  

\begin{figure}[ht]
\hspace*{\fill}
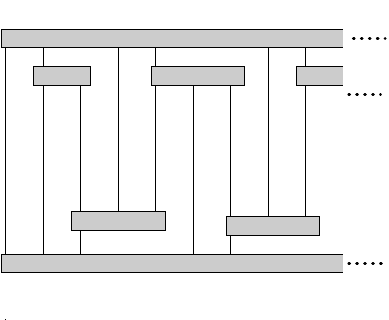
\hspace*{\fill}
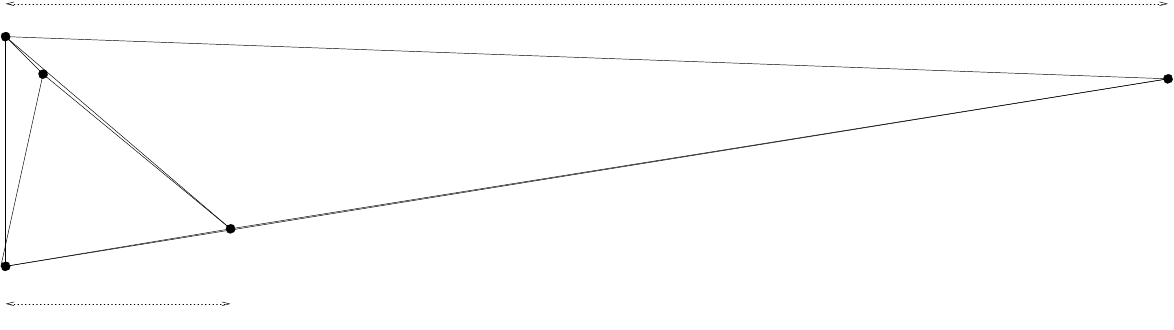
\hspace*{\fill}
\caption{A flat visibility representation for which the corresponding
straight-line drawing has exponential width.  Vertices are numbered
in the order in which they are processed.  }
\label{fig:bad_example}
\end{figure}

\begin{lemma}
For $h\geq 4$, the width of the drawing obtained with 
Theorem~\ref{thm:VR_SL} is $O((h-2)^{n-3})$.
\end{lemma}
\begin{proof}
Define two recursive functions as $W(2)=W'(2)=0$, 
$W'(i)=1+(h-2)W'(i-1)$ and $W(i)=1+(h-1)W'(i-1)$ for $i\geq 3$. 
%
%
We will show that for $i\geq 4$, 
any point $p$ of the drawing $\Gamma_i'$
has $x$-coordinate at most $W(i)$, and if $p$
is not on the first or last row, then it has
$x$-coordinate at most $W'(i)$.
Observe that $W'(i)=1+(h-2)+\dots+(h-2)^{i-3}\in O((h-2)^{i-3})$
and $W(i)=W'(i-1)+W'(i)\leq 2W'(i)$; hence $W(n)\in O((h-2)^{n-3})$ as desired.

For the base case, we have two cases.  If two of $\{v_1,v_2,v_3\}$
are in different rows, then two of $\{v_1,v_2,v_3\}$ have
$x$-coordinate 0 and the third has $x$-coordinate at most 1, hence
the claim holds for $\Gamma_3'$ since $W(3)=W'(3)=1$.  If $v_1,v_2,v_3$
are all in the same row, then they have $x$-coordinates 0,1,2.  Vertex
$v_4$ is either also placed in this row (and then has $x$-coordinate 3)
or it is in a different row (and then has $x$-coordinate 0.)  Either way,
all points in $\Gamma'_4$ then have $x$-coordinate at most
$3\leq W'(4)\leq W(4)$.  This shows the base case.

For the induction step, we distinguish cases on how the $x$-coordinate $X(v_i)$
of $v_i$ was determined:
\begin{itemize}
\item Assume first that $X(v_i)$ was determined via $X(v_i) = \lfloor X
\rfloor +1$, where $X$ is the $x$-coordinate of some point $p$ in 
$\Gamma'_{i-1}$ in the row of $v_i$.  By induction we know that
$X\leq W(i-1)$, and so $X(v_i)\leq 1+X \leq 1+W(i-1) \leq 1+2W'(i-1)\leq W(i)$
by $h\geq 3$.
If $v_i$ is not in the first or last row, then neither is $p$, so
$X\leq W'(i-1)$ and $X(v_i)\leq 1+W'(i-1) \leq W'(i)$.

\item Assume now that $X(v_i)=\lfloor X_g \rfloor +1$, where $X_g$
is the $x$-coordinate of the intersection of the row of $v_i$ with
a line $\ell$ through some predecessor $v_g$
of $v_i$ and some point $p$ of drawing $\Gamma'_{i-1}$.
Assume as in Figure~\ref{fig:transform} that
$y(v_g)\geq y(p)$; the other case is
similar.  
This implies $y(v_g)\geq y(v_i)$ as well, otherwise $p$ would
not have been an obstruction for the edge $(v_g,v_i)$.  

If $X(p)\leq X(v_g)$, then $X_g\leq X(v_g)$, therefore
$X(v_i)\leq X(v_g)+1$ satisfies the bound as in the first case.
If $p$ is in the bottommost row, then $X_g\leq X(p)$, therefore
$X(v_i)\leq X(p)+1$ satisfies the bound as in the first case.
Finally assume $X(p)> X(v_g)$ and $p$ is not in the bottommost
row, hence $X(p)\leq W'(i-1)$.
Now
$$X_g = X(v_g) + (y(v_g)-y(v_i)) \frac{X(p)-X(v_g)}{y(v_g)-y(p)}
\leq (y(v_g)-y(v_i)) X(p)
\leq (y(v_g)-y(v_i)) W'(i-1).$$
If $y(v_g)-v(v_i)\leq h-2$ then $X(v_i)\leq 1+X_g \leq 
1+(h-2)W'(i-1)=W'(i)$.  Otherwise $v_i$ is in the
bottommost row (and $v_g$ in the top row), and 
$X(v_i)\leq 1+X_g \leq 1+(h-1)W'(i-1)=W(i)$ as desired.
\end{itemize}
\end{proof}

We note here that $h\geq 4$ is needed only for $1+(h-2)+\dots+(h-2)^{n-3}
\in O((h-2)^{n-3})$; much the same proof shows that the height is $O(n)$
for $h=3$.

\section{$y$-monotone flat orthogonal drawings to flat visibility 
representations}

\begin{theorem}
\label{thm:o2VR}
\label{th:o2VR}
\label{thm:OD_VR}
Any flat $y$-monotone orthogonal drawing $\Gamma$ can be transformed into a 
flat visibility representation $\Gamma'$ with the same $y$-coordinates
and the same left-to-right orders in each row.
\end{theorem}
\begin{proof}
First, expand every
vertex $v$ to the left and right until it covers all bends (if any)
of edges that attach horizontally at $v$.  Since $v$ has height 1,
there is at most one edge $e$ each on the left and right side of $v$,
and the expansion of $v$ covers only space previously used by $e$,
hence creates no overlap.

Now we arrive at a drawing where all edges that have bends attach
vertically at their endpoints.  
Let $e$ be an edge with bends (if there
is none we are done.)  Since $e$ is drawn $y$-monotone, it attaches
at the top of one endpoint and the bottom of the other endpoint,
and the only way it can have bends is to have a right turn followed
by a left turn or vice versa.  Thus, $e$ has a ``zig-zag''.  It is
well known that such a zig-zag can be removed by transforming the
drawing as follows (see also Figure~\ref{fig:remove_Z}):  
Extend the ends of the zig-zag upward and downward
to infinity, and then shift the two sides of the resulting separation
apart until the two rays of the zig-zag align.  See for example
\cite{BLPS09} for details.  This operation adds width, but no height.
Applying this to all edges that have bends gives A
visibility representation.   
\end{proof}

\begin{figure}[ht]
\hspace*{\fill}
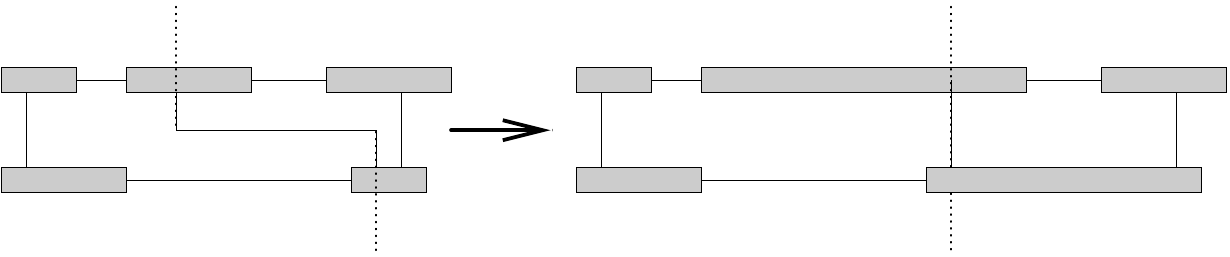
\hspace*{\fill}
\caption{Removing a zig-zag by shifting parts of the drawing rightwards.}
\label{fig:remove_Z}
\end{figure}

\subsection{Width considerations}

Our construction may increase the width quite a bit, but mostly with
columns that are {\em redundant}: They contain neither a vertical
edge, nor are they the only column of a vertex.  A natural post-processing
step is to remove redundant columns.  We then obtain small width.
In fact, the width is small for {\em any} visibility representation.

\begin{lemma}
Any visibility representation of a connected graph has width at most $\max\{m,n\}$ after deleting redundant columns.
\end{lemma}
\begin{proof}
Let $m_h$ and $m_v$ be the number of edges drawn horizontally and vertically.
Let $V_h$ be the vertices without incident vertical edge.
Then the width is at most $m_v+|V_h|$.   This shows the claim if $m_v=0$.
If $m_v>0$, then let $v$ be a vertex not in $V_h$, 
pick an arbitrary spanning tree $T$ and root it at $v$.
For any vertex $w\in V_h$, the edge from $w$ to its parent in $T$
must be horizontal by definition of $V_h$.  Hence there are
at least $|V_h|$ horizontal edges,
and the width is at most $m_v+m_h= m$.
\end{proof}

\section{Poly-line drawing to flat orthogonal drawing}
\label{se:SL_VR}

\begin{theorem}
\label{thm:PL_OD}
\label{th:PL2OD}
Any poly-line drawing $\Gamma$ can be transformed into a 
flat orthogonal drawing $\Gamma'$ with the same $y$-coordinates
and the same left-to-right orders in each row.
$\Gamma'$ is $y$-monotone if $\Gamma$ was.
\end{theorem}
\begin{proof}
We first transform $\Gamma$ into an  {\em $h$-layer drawing}, i.e.,
a straight-line drawing where all edges are horizontal or connect
adjacent rows.  We do this by
inserting {\em pseudo-vertices} (i.e., subdivide edges) at bends and
whenever a segment of an edge crosses a row.  (We 
allow non-integral $x$-coordinates for pseudo-vertices.)

For each row $r$, let $w_1,\dots,w_k$ be the vertices (including
pseudo-vertices) in $r$ in left-to-right order.  In $\Gamma'$,
replace each $w_i$ by a box of width 
$\max\{1,\deg^{up}(w_i),\deg_{down}(w_i)\}$, where $\deg^{up}(w_i)$
and $\deg_{down}(w_i)$ are the number of neighbours of $w_i$ with
larger/smaller $y$-coordinate.  Place these boxes in row $r$ in
the same left-to-right order.

Each horizontal edge is drawn horizontally in $\Gamma'$ as well.
Each non-horizontal edge connects
two adjacent rows since we inserted pseudo-vertices.  
Connect the edges between two adjacent rows using VLSI channel routing
(see e.g.~\cite{Len90}), using two bends per edge and lots of new rows
(with non-integer $y$-coordinates) that contain
horizontal edge segments and nothing else.

\begin{figure}[ht]
\hspace*{\fill}
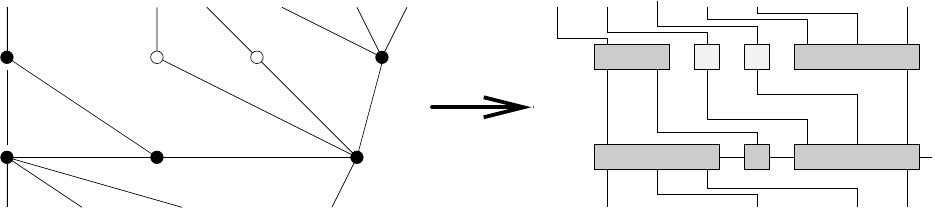
\hspace*{\fill}
\caption{Converting a poly-line drawing to an orthogonal drawing.
Pseudo-vertices are white.}
\label{fig:PL2OD}
\end{figure}

Now bends only occur at zig-zags; remove these as in the proof of 
Theorem~\ref{thm:o2VR}.  This empties all rows except those
with integer coordinates and gives the desired height and
a flat visibility representation of the graph with pseudo-vertices.
Any pseudo-vertex can now be removed and replaced by a bend if
needed.
\end{proof}

Note that the visibility representation obtained as part of the
proof has width at most $\max\{n+p,m+p\}$, where $p$ is the number
of pseudo-vertices inserted.  An even better bound can be obtained
by observing that any pseudo-vertex that is not at a bend in $\Gamma$  
will receive two incident vertical segments and hence can be removed
in the visibility representation.
So the width is 
at most $\max\{n+b,m+b\}$, where $b$ is the number of bends
in $\Gamma$.

\section{Flat orthogonal drawings to poly-line drawings}
\label{se:VR_PL}

Combining the previous theorems, it is easy to see that
any flat orthogonal drawing can be converted to a poly-line
drawing of the same height:  First convert it to a visibility
representation, then convert it to a straight-line drawing,
and then interpret the result as a poly-line drawing.
However, since this involves Theorem~\ref{thm:VR_SL}, the
width might grow exponentially.  We now give a simple direct
proof of this transformation that shows that it can be done
while keeping the width small.

\begin{theorem}
Any flat orthogonal drawing $\Gamma$ can be transformed into a 
poly-line drawing $\Gamma'$ with the same $y$-coordinates
and the same left-to-right orders in each row.
Moreover, $\Gamma'$ has no more width than $\Gamma$,
and it is $y$-monotone if $\Gamma$ was.
\end{theorem}
\begin{proof}
First subdivide edges at all bends and all vertical edge-segments 
with pseudo-vertices so that
any vertical edge connects two vertices in adjacent layers.  
Apply the algorithm in Theorem~\ref{thm:VR2SL} to find a straight-line drawing 
of the resulting graph; 
removing the pseudo-vertices then gives the desired poly-line drawing.
All properties are easily verified, except for the width.  Observe
that when applying the construction of Theorem~\ref{thm:VR2SL},
for any vertex $v_i$ the predecessors are in the same or in adjacent
rows.  Hence all lines from $v_i$ to predecessors are unobstructed,
and $v_i$ can simply be placed in the leftmost free position of its row.
Hence the width of $\Gamma'$ is the maximal number of vertices or 
pseudo-vertices in a row, which is no more than the width
of $\Gamma$ since $\Gamma$ is orthogonal.
\end{proof}

\section{Applications}

We give a few applications of the results in this paper.

\subsection{Drawing graphs with small height}

The {\em pathwidth} $pw(G)$ of a graph $G$ is a graph
parameter that is related to heights of planar graph drawing:
any planar graph that has a straight-line drawing of height
$h$ has pathwidth at least $h$ \cite{FLW03}.  But not all graphs
with pathwidth $h$ have a drawing of height $O(h)$.  Our transformations
show that such heights do exist for outer-planar graphs:

\begin{corollary}
Any outer-planar graph $G$ has a straight-line drawing of height
$O(pw(G))$.
\end{corollary}
\begin{proof}
By a result of Babu et al.~\cite{BBC+12}, we can add edges to $G$
to obtain a maximal outerplanar graph $G'$ with pathwidth in $O(pw(G))$.
In particular, $G'$ is 2-connected and hence by 
\cite{Bie-WAOA12} it has a flat visibility representation of
height at most $4pw(G') \in O(pw(G))$.  By Theorem~\ref{thm:VR_SL}, therefore
$G'$ (and with it $G$) has a straight-line drawing of height $O(pw(G))$.
\end{proof}

Recall that outerplanar graphs have constant treewidth and hence
pathwidth $O(\log n)$, so any outerplanar graph has a straight-line
drawing of height $O(\log n)$.  

In a similar fashion, any graph drawing algorithm that produces
drawings of small height in one of our models produces, with our
transformations, graph drawings of small heights in all other models.
We give one more example:

\begin{corollary}
Any 4-connected planar graph has a visibility representation
of height at most $\lfloor n/2 \rfloor$.
\end{corollary}
\begin{proof}
It is known that any 4-connected planar graph has
a straight-line drawing 
where the sum of the width and height is at most $n$ \cite{MNN06}.
Therefore, after possible rotation, the height is at most
$\lfloor n/2 \rfloor$,
and with 
Theorem~\ref{thm:OD_VR} and \ref{thm:PL_OD} 
we get a flat visibility representation of height $\lfloor n/2 \rfloor$.
\end{proof}

The best previous bound on the height of visibility representations
of 4-connected planar graphs was $\lceil \frac{3n}{4} \rceil$
\cite{HWZ12}.

\subsection{Integer programming formulations}

In a recent paper, we developed integer program (IP) formulations
for many graph drawing problems where vertices and edges are
represented by axis-aligned boxes \cite{Bie-GD13}.  By adding
some constraints, one can force that edges degenerate to line segments
and vertices to horizontal line segments.  In particular, it is easy to 
create an IP that expresses ``$G$ is drawn as a flat visibility
representation'', using $O(n^3)$ variables and constraints.
With the transformations given in this paper, we can then use IPs
for many other graph drawing problems. 
%
%
The following result (based on Theorem~\ref{thm:VR_SL}, 
\ref{thm:OD_VR}, 
\ref{thm:PL_OD}) is crucial: 

\begin{corollary}
A planar graph $G$ has a planar straight-line drawing of height $h$
if and only if it has a flat visibility representation of height $h$.
\end{corollary}

It is very easy to encode the height in the IP formulations
of \cite{Bie-GD13}. 
We therefore have:

\begin{corollary}
There exists an integer program with $O(n^3)$ variables and constraints
to find the minimum height of a planar straight-line drawing of a graph $G$.
\end{corollary}


A directed acyclic graph has an
{\em upward drawing} if it has a planar straight-line drawing such
that for any directed edge $v\rightarrow w$ the $y$-coordinate of $v$
is smaller than the $y$-coordinate of $w$.  Testing whether a graph
has an upward drawing is NP-hard \cite{GT01}.  There exists a way
to formulate ``$G$ has an upward drawing'' as either IP
or as a Satisfiability-problem, using partial orders on
the edges and vertices \cite{CZ-GD12}.  Our transformations give
a different way of testing this via IP:

\begin{lemma}
A directed acyclic graph has an upward planar drawing if and only
if it has a visibility representation where all edges are vertical
lines, with the head above the tail.
\end{lemma}
\begin{proof}
Given a straight-line upward drawing, we can transform it into the
desired visibility representation using Theorems~\ref{thm:OD_VR}
and \ref{thm:PL_OD}.
Since $y$-coordinates are unchanged, any edge is necessarily drawn
vertical with the head above the tail.  Vice versa, given such a visibility
representation, we can transform it into a flat visibility representation
simply by replacing boxes of positive height by horizontal segments; this
leads to no conflict since there are no horizontal edges.  Then apply
Theorem~\ref{thm:VR_SL}; this gives an upward drawing since $y$-coordinates
are unchanged.
\end{proof}

It is easy to express ``edge $v\rightarrow w$ must be drawn vertically,
with the head above the tail'' as constraints in the IP 
for visibility representations defined in \cite{Bie-GD13}.  We therefore
have:

\begin{corollary}
There exists an integer program with $O(n^3)$ variables and
constraints to test whether a planar graph has an upward planar drawing.
Moreoever, the same integer program also finds the minimum-height
upward drawing.
\end{corollary}

\section{Conclusion and open problems}
\label{se:open}

In this paper, we studied transformations between different types
of graph drawings, in particular between straight-line drawings
and flat visibility representations.  We demonstrated applications
of these results, especially for drawings of small heights, and
upward drawings.

We have not been able to create transformations that start with an
arbitrary (i.e., not necessarily flat) visibility representation and
turn it into a straight-line drawing of approximately the same height.
Does such a transformation exist?

Another open problem concerns the width, especially for the transformation
from flat visibility representations to planar straight-line drawings.
Is it possible to make the width polynomial if we may change the
$y$-coordinates while keeping the height asymptotically the same?

\section*{Acknowledgments}

Research partially supported by NSERC and by the Ross and Muriel
Cheriton Fellowship.  Some of the results in this paper appeared
in \cite{Bie-WAOA12}.

\bibliographystyle{plain}
\bibliography{../../bib/full,../../bib/papers,../../bib/gd}

\end{document}